\documentclass[aps,twocolumn,superscriptaddress,showpacs]{revtex4-1}

\usepackage{amsmath}
\usepackage{amsthm} 
\usepackage{graphicx}
\usepackage{epsf}
\usepackage{color}
\usepackage[normalem]{ulem}

\newcommand{\cA}{\mathcal{A}}
\newcommand{\cS}{\mathcal{S}}
\newcommand{\rhoS}{\rho_\mathrm{S}}
\newcommand{\rhoref}{\rho_\mathrm{B}^\mathrm{(ref)}}
\newcommand{\HSB}{H_\mathrm{SB}}
\newcommand{\HS}{H_\mathrm{S}}
\newcommand{\HB}{H_\mathrm{B}}
\newcommand{\tr}[1]{\mathrm{tr}{\left\{#1\right\}}}
\newcommand{\rhoME}{\rho^{(\mathrm{ME})}}
\newcommand{\AME}{\cA^{(\mathrm{ME})}}
\newcommand{\trB}{\mathrm{tr}_\mathrm{B}}
\newcommand{\mapAref}{\cA^\mathrm{(ref)}}
\newcommand{\rhoBref}{\rho_\mathrm{B}^{(\mathrm{ref})}}
\newcommand{\rhoSa}{\rho_{\mathrm{S},1}}
\newcommand{\rhoSb}{\rho_{\mathrm{S},2}}
\newcommand{\rhoBth}{\rho_\mathrm{B}^{(\mathrm{th})}}
\newcommand{\upe}{\mathrm{e}}
\newcommand{\upd}{\mathrm{d}}
\newcommand{\ket}[1]{|#1\rangle}
\newcommand{\bra}[1]{\langle#1|}
\newcommand{\mapA}{\cA_\beta^{(\mathrm{ME})}}
\newcommand{\ordO}{\mathcal{O}}
\newcommand{\rhoSB}{\rho_\mathrm{SB}}
\newcommand{\uket}{|\!\!\uparrow\rangle}
\newcommand{\ubra}{\langle\uparrow\!\!|}
\newcommand{\dbra}{\langle\downarrow\!\!|}
\newcommand{\dket}{|\!\!\downarrow\rangle}
\newcommand{\upi}{\mathrm{i}}
\newcommand{\rhoB}{\rho_\mathrm{B}}

\newtheorem*{proposition*}{Proposition}
\newtheorem{corollary}{Corollary}

\begin{document}

\title{The initial system-bath state via the maximum-entropy principle}
\author{Jibo \textsc{Dai}}
\affiliation{Centre for Quantum Technologies, National University of
  Singapore, Singapore 117543}
\author{Yink Loong \textsc{Len}}
\affiliation{Centre for Quantum Technologies, National University of
  Singapore, Singapore 117543}
\author{Hui Khoon \textsc{Ng}}
\affiliation{Centre for Quantum Technologies, National University of
  Singapore, Singapore 117543} 
\affiliation{Yale-NUS College, Singapore 138609}
\affiliation{MajuLab, CNRS-UNS-NUS-NTU International Joint Research Unit, UMI 3654, Singapore}
\date{Manuscript draft of \today}

\begin{abstract}
The initial state of a system-bath composite is needed as the input for prediction from any quantum evolution equation to describe subsequent system-only reduced dynamics or the noise on the system from joint evolution of the system and the bath. The conventional wisdom is to write down an uncorrelated state as if the system and the bath were prepared in the absence of each other; yet, such a factorized state cannot be the exact description in the presence of system-bath interactions. Here, we show how to go beyond the simplistic factorized-state prescription using ideas from quantum tomography: We employ the maximum-entropy principle to deduce an initial system-bath state consistent with the available information. For the generic case of weak interactions, we obtain an explicit formula for the correction to the factorized state. Such a state turns out to have little correlation between the system and the bath, which we can quantify using our formula.
This has implications, in particular, on the subject of subsequent non-completely-positive dynamics of the system. Deviation from predictions based on such an almost uncorrelated state is indicative of accidental control of hidden degrees of freedom in the bath.
\end{abstract}

\maketitle

\section{Introduction}
Time-evolution problems require the specification of initial conditions for the prediction of future behavior.
In quantum evolution, these initial conditions come in the form of the initial state of the quantum system in question. For open quantum systems, subsequent dynamics involve not just the system of interest, but also the bath---the coupling to which gives non-unitary evolution of the system---and the joint initial system-bath state is of relevance. How does one write down a reasonable initial system-bath state, given the inability to perform tomography on the bath, the defining quality of which is that of inaccessibility and uncontrollability?

Usually, one argues that the system is well isolated from the bath: Only when the coupling is weak does this split into system and bath make good physical sense. The system is then viewed as having been prepared in some state $\rhoS$ independently of the bath. The bath, usually much larger than the system and hence suffering negligible influence from the system, is taken to be in some reference state $\rhoref$ (e.g., the thermal state) as if the system were absent. This reasoning gives the factorized state $\rhoS\otimes \rhoref$ as the initial system-bath state. One arrives at the same state by alternatively imagining that the system-bath interaction is zero for time $t<0$, and is ``turned on" only at $t=0$. 

Practically, the isolation of the system from the bath can never be complete; e.g., one can never decouple from the electromagnetic vacuum. One expects the initially perfectly uncorrelated situation to be but an approximation, albeit often a good one. As we advance towards more and more sensitive quantum precision measurements, we are less and less justified in ignoring these system-bath correlations. The presence of initial correlations also raises interesting questions that have triggered much ongoing discussion regarding the description of subsequent system-only dynamics: the lack of a well-defined reduced dynamics \cite{Pechukas:94,Alicki95,Pechukas95}; the connection between non-completely-positive (non-CP) dynamics and correlated initial joint system-bath states \cite{Pechukas:94,Alicki95,Pechukas95,Stelmachovic01,Jordan04,Shaji05,Buscemi14} (see Ref.~\cite{Buscemi14} for further references to the many works discussing the consequences of non-CP dynamics in a variety of settings); more specifically, the relationship between completely-positive (CP) dynamics and initial states with vanishing discord \cite{Rodriguez-Rosario08,Shabani+Lidar:09,Brodutch13}; attempts to define reduced dynamics in the presence of initial system-bath correlations under restricted conditions \cite{Royer96,Carteret08,Shabani09,Buscemi14,Dominy15}; etc.
The non-CP nature of the reduced system-only dynamics resulting from initial correlations also have close links with non-Markovian dynamics of the system, another subject of vigorous recent discussion and creative exploitation of the resulting effects (see Refs.~\cite{Rivas14} and \cite{Breuer15} for comprehensive reviews on the subject). An initial correlated system-bath state may also have a bearing on conclusions in the discussion of topics like quantum fluctuation theorems \cite{Hanggi15}.

In many of these articles, the initial system-bath state is unspecified, taken to be some fixed but unknown state---unknown because one is unable to perform state tomography on the bath. One should, however, be able to write down a reasonable initial system-bath state. Afterall, a state is \emph{our}---the experimenters'---best description of the physical situation. The system-bath state should thus express what we \emph{know} about the system and bath, incorporating any prior knowledge of the experimental circumstances, as well as any characterization data of the initial system state. Furthermore, it should express our \emph{lack} of knowledge of the precise identity of the bath state, stemming from our inability to control the microscopic state of the bath, the latter property forming the basis for the dichotomy into bath versus system. As such, the framework of state estimation \cite{paris2004quantum} provides a solid foundation for writing down a joint system-bath state. In particular, the approach of maximum-entropy (ME) state estimation, well justified in the writings of Jaynes \cite{Jaynes:1957,Jaynes:1982,JaynesBook} and used more recently in quantum state tomography \cite{Buzek00,Teo11,Teo12}, offers a concrete and philosophically satisfactory way of representing both our knowledge as well as lack thereof for the system-bath composite.

In this article, we examine this ME route towards a reasonable initial system-bath state. The ME approach identifies as the best guess for the state, the statistical operator that maximizes the entropy in the space constrained by our (partial) knowledge from available data, with a nonzero entropy arising from the incompleteness of the knowledge, i.e., our ignorance. 
In Sec.~II, we set the stage for discussion by summarizing the knowledge of the system-bath composite obtainable in a typical experimental situation. In Sec.~III, we explain the use of the ME principle for constructing a joint system-bath state. We also point out the relationship to the assignment map approach to reduced system-only dynamics. In Sec.~IV, we examine the case of a weak system-bath interaction at fixed temperature, the generic experimental situation and the case of most relevance in the discussion of initial system-bath correlations. We obtain a formula for the correction to the factorized state, which forms our main result. This formula is used to derive a bound on the initial system-bath correlations for an arbitrary weak interaction Hamiltonian. We discuss implications on the subsequent system-only dynamics, and illustrate with some examples. In Sec.~V, we mention some related issues, and conclude in Sec.~VI.

\section{Our knowledge}

Controlled quantum experiments begin with the preparation of the system in some target state at time $t=0$, before letting it evolve naturally or according to a sequence of gates for $t>0$. That the system is well-initialized into the desired state is verified via tomography: The system is repeatedly prepared under the same experimental conditions, and the copies are measured using a chosen tomography strategy; the data are used to estimate the prepared state. We phrase this estimate---what we know of the system---in terms of estimates of the expectation values $\{o_i\}_{i=1}^K$ of a set of operators $\{O_i\}_{i=1}^K$ on the system state. In a typical tomographic setting, $\{O_i\}$ is the set of tomographic outcomes (or POVM) and the $o_i$ is the probability that the detector for outcome $i$ registers a click.

The set $\{O_i\}$ can be complete in that knowing their expectation values identifies exactly one density operator $\rhoS$ for the system state; or it can be incomplete, and allows for a restricted set of $\rhoS$s consistent with the data. We will assume that $\{O_i\}$ is complete, although it is straightforward to incorporate the case of an incomplete set using methods from analyzing incomplete tomographic data.
We will also assume that we have good certainty of our estimates for the $o_i$s. Cases where uncertainty arises because of limited data have been studied in the literature and can be incorporated using the concept of estimator regions \cite{Christandl+1:12,Blume-Kohout12,Shang13}. We will, however, refrain from putting in such complications, and only point out that one should in general include knowledge about the \emph{quality} of the available information about the system, not just \emph{what} information is available.

What about our knowledge of the bath? By ``bath", we refer to the part of the relevant physical system not amenable to microscopic control. We can at best determine only macroscopic properties, like temperature, pressure, volume, etc. By ``relevant", we mean that we include only those degrees of freedom that couple sufficiently (though weakly) to our controlled system to affect its dynamics; the ``rest of the world" can be ignored at the level of precision detectable in the experiment. We possess no microscopic knowledge of the bath, having no ability to measure the bath in the same way as we could perform tomography on the system. We do, however, assume that we have correctly identified those physically relevant degrees of freedom, e.g., the cavity modes that are resonant or near-resonant to a few-level atom (the idealized system) inside the cavity; the surrounding environment of nuclear spins around a nitrogen-vacancy-centre spin site; etc. 

With this identification also comes the physical model for the coupling between the system and the bath, which we can phrase in terms of an interaction Hamiltonian $\HSB$. More general coupling, e.g., non-unitary evolution via a joint master-equation approach, or a positivity- and trace-preserving map, is possible, but for our current discussion, we will restrict ourselves to the simplest joint unitary evolution. Alongside this $\HSB$, we can also include a system-only Hamiltonian $\HS$, and a bath-only Hamiltonian $\HB$.

In summary, our knowledge of the system-bath composite at time $t=0$ can be collected as follows:
\begin{enumerate}
\renewcommand{\itemsep}{1pt}
\item information about the microscopic system state $\rhoS$ from knowing $o_i=\tr{\rhoS O_i}$, $i=1,2,\ldots, K$;
\item the relevant bath degrees of freedom;
\item the full Hamiltonian $H=\HS+\HB+\HSB$;
\item macroscopic properties of the system and the bath. 
\end{enumerate}
Item 1 assumes repeatability in the preparation of the system state; items 2 and 3 assume a reasonable model for the system and bath, based on the experimenter's understanding of the physical situation; the last requires the laboratory measurement of only macroscopic properties of the system-bath composite. 

The above knowledge is the basic starting point for any joint system-bath evolution investigation, and hence is the right place to begin our construction of a joint system-bath state consistent with this information.

\section{\label{sec:3}The maximum-entropy state}
Without the ability to perform tomography on the bath to determine its microscopic state, we have only incomplete information about the overall system-bath state at time $t=0$. The state that we write down must hence reflect this partial ignorance. The use of ME methods in such situations follows a long tradition dating back to the beginnings of statistical physics. For statistical inference problems, Jaynes \cite{Jaynes:1957,Jaynes:1982,JaynesBook} argued lucidly for the viewpoint of maximum entropy as the most rational approach in the face of incomplete information, and Ref.~\cite{Buzek00} introduced the ME approach to estimate the density operator from known expectation values of  observables. ME ideas are utilized more recently in quantum state estimation with informationally incomplete measurements \cite{Teo11,Teo12}.

The ME state $\rhoME$ for the system-bath composite is the state that attains 
\begin{equation}\label{ME_State_def}
\cS_\mathrm{max}\equiv\cS{\bigl(\rhoME\bigr)}=\max_{\rho\in\Sigma} S(\rho), 
\end{equation}
where $\cS(\rho)=-\tr{\rho\log\rho}$, is the von Neumann entropy. Here, $\Sigma$ is the set of all system-bath states consistent with our partial knowledge of the system and the bath.
This maximization over $\Sigma$ can be implemented as an unconstrained maximization over all states using the Lagrange-multiplier method, so that we instead maximize
\begin{eqnarray}
\tilde\cS(\rho)&\equiv&\cS(\rho)-\sum_i\lambda_i\bigl(\tr{\rho O_i}-o_i\bigr)\nonumber\\
\label{entropy_with_constraint1}&&\hspace*{0.6cm} -\nu\bigl(\tr{\rho}-1\bigr)-\mu\bigl(\tr{G(\rho)}-g\bigr) 
\end{eqnarray}
over all system-bath $\rho$.
Here, $\lambda_i$s are the Lagrange multipliers for the system information, the $\nu$ term enforces the unit-trace constraint for $\rho$, and $\mu(\tr{G(\rho)}-g)$ is symbolic for the Lagrange multiplier term that accounts for any additional knowledge one may have. The ME state is obtained from the stationarity condition, i.e., setting 
\begin{eqnarray}
\delta \tilde S(\rho)&\equiv&\tilde S(\rho+\delta\rho)-\tilde S(\rho)\\
&=&-\tr{\delta\rho{\left(\text{log}\,\rho+1+\nu+\sum_i\!\lambda_iO_i+\mu\frac{\partial G}{\partial \rho}\right)}}\nonumber
\end{eqnarray}
to zero for arbitrary infinitesimal $\delta\rho$. This requires
\begin{equation}
\rho=\exp{\left(-1-\nu-\sum_i\lambda_iO_i-\mu\frac{\partial G}{\partial \rho}\right)}.
\end{equation}
The trace-1 constraint associated with $\nu$ can be taken care of immediately, so that we have the ME state as
\begin{equation}\label{eq:rhoME}
\rhoME=\exp{\left(-\Lambda-\mu\frac{\partial G}{\partial \rho}\right)}/\tr{\ldots},
\end{equation}
where $\Lambda\equiv \sum_i\lambda_iO_i$. Here, $\tr{\ldots}$---meaning the trace of the numerator when appearing in the denominator of a fraction---takes care of the unit-trace constraint associated with $\nu$. Note that there are proposals to use Bayesian corrections to the ME scheme in the determination of the Lagrange multipliers, when the sample sizes are small \cite{Rau:2011}. Here, we assume that we have measured enough copies so that the $\lambda_i$s can be accurately determined.

For each experimental situation, one can view the ME strategy as an \emph{assignment map} $\AME$ that takes a system-only state $\rhoS$ to the corresponding system-bath state $\rhoME$. Assignment maps were first discussed in Refs.~\cite{Pechukas:94,Alicki95,Pechukas95} as a one-to-one association between system-only and system-bath states to permit discussion of reduced system-only dynamics. Concretely, with the ME assignment map, one can describe the reduced system dynamics as a dynamical map
\begin{eqnarray}\label{eq:Phi}
\rhoS(t)&=&\Phi_t(\rhoS),\\
\textrm{with} \quad\Phi_t(\cdot)&\equiv&\trB\{U(t,0)\AME(\cdot)U(t,0)^\dagger\}.\nonumber
\end{eqnarray}
Here, $\rhoS$ written without any time argument is the initial system state verified by tomography on the system, and $U(t,0)$ is the system-bath joint evolution operator for $H$. 

For any assignment map $\cA$, Refs.~\cite{Pechukas:94,Alicki95,Pechukas95} considered three properties: (a) linear [preserves mixtures, i.e., $\cA(p\rhoS+(1-p)\sigma_\mathrm{S})=p\cA(\rhoS)+(1-p)\cA(\sigma_\mathrm{S})$], (b) positivity preserving [$\cA(\rhoS)\geq 0$ for all $\rhoS\geq 0$], and (c) consistent [$\trB\{\cA(\rhoS)\}=\rhoS$]. Pechukas \cite{Pechukas:94} showed that the only assignment map that possesses all three properties is the factorized map, $\mapAref(\rhoS)=\rhoS\otimes \rhoBref$, for some fixed reference bath state $\rhoBref$. Our ME assignment map $\AME$, generally not of this factorized structure except in special circumstances (see below), satisfies (b) and (c), but not (a). That it does not preserve mixtures should not bother us: When preparing a mixed state $\rhoS$, one may follow a procedure that produces the state directly, or one may get $\rhoS$ as a probabilistic mixture of two separately prepared states $\rhoSa$ and $\rhoSb$. There is no reason why the $\rhoME$ states for the two different preparations---corresponding to two different physical scenarios---of $\rhoS$ should coincide. According to the ME philosophy, the actual preparation procedure used should form part of the prior information entering the constraints. The two different preparations constitute different prior information and should hence lead to different $\rhoME$ states.

As an example of the ME strategy, we consider the typical situation where the only controlled macroscopic parameter is the temperature $T$. Such a constraint appears in Eqs.~\eqref{entropy_with_constraint1} and \eqref{eq:rhoME} with $\mu= \beta\equiv 1/kT$ ($k$ is the Boltzmann constant), $G(\rho)=\rho H$, and $g$ is the average energy. The corresponding ME state is
\begin{equation}\label{eq:rhoME1}
\rhoME=\exp(-\Lambda-\beta H)/\tr{\ldots}.
\end{equation}
Let us examine three cases in this context.

\textit{Case 1.} Suppose the bath is in thermal equilibrium, and the system is exposed to the bath only at time $t=0$. This describes many experimental situations where the system is prepared in a different part of the experimental setup and then moved into the main experimental chamber at time $t=0$. The system-bath state---upon introducing the system into the final chamber---should hence be the product $\rhoS\otimes \rhoBth$, where $\rhoBth\equiv \exp(-\beta\HB)/\tr{\ldots}$ is the bath-only thermal state at temperature $T$. Indeed, maximizing entropy subjected to the tomography constraints on the system and the bath temperature---incorporated using $\mu g(\rho)$---yields the expected state. The product structure is clear as $\tilde\cS(\rho)$ involves constraints on system and bath separately, and $\rhoBth$ emerges as the usual thermal state from maximizing entropy for fixed temperature.
The appropriate ME assignment map for this physical situation is hence of the factorized $\mapAref$ variety,
\begin{equation}
\cA_0(\rhoS)=\rhoS\otimes\rhoBth,
\end{equation}
with $\rhoBref=\rhoBth$.

\textit{Case 2.} The situation gets more complicated when the system is prepared in constant contact with the bath, and the combined system-bath is at some temperature $T$. In this case, there is little support for an initially uncorrelated system-bath state $\rhoS\otimes \rhoref$, although this is commonly used in the literature even for such cases. The preparation of the system state can sometimes be accomplished very rapidly, but even so, the coupling between system and bath is always on and can modify the outcome of the preparation procedure. Many schemes employ slow or long processes that gradually move the system population into the desired state, and during this transfer, the interaction $\HSB$ can be thought of as continuously keeping the system in thermal equilibrium with the bath. In such cases, the coupling to the bath is accounted for through the use of $\rhoME$ as the initial system-bath state for the appropriate $T$ and $H$. This case is examined carefully in the next Section, assuming a weak system-bath coupling.

\textit{Case 3.} One might also have no macroscopic information altogether.
In this case, there is no $\mu G(\rho)$ term, and the ME state is $\rhoME=\upe^{-\Lambda}/\tr{\ldots}$.
Since $\Lambda$ is a system-only operator, this simplifies to $\rhoME=\rhoS\otimes \frac{1}{d_\text{B}}\boldsymbol{1}_\mathrm{B}$, where $d_\text{B}$ is the dimension of the bath. $\rhoS$ is the usual state one would have deduced from the tomography data, with $\lambda_i$s determined only by the constraint summarized as $\trB\{\rho\}=\rhoS$. Again, the ME assignment map here is of the factorized type $\mapAref$, with $\rhoBref=\frac{1}{d_\text{B}}\boldsymbol{1}_\mathrm{B}$. The maximally mixed state on the bath appears naturally from this approach, expressing our complete ignorance of the bath state.

\section{Weak system-bath coupling at fixed temperature}\label{sec:weakCoup}
Let us examine more closely the often-encountered situation of constant contact between the system and the bath, both held at some temperature $T$ (Case 2). We assume a weak coupling between the system and the bath. This always-on weak coupling scenario is the case of most relevance in the discussion of initial system-bath correlations. Below, we derive an explicit formula for $\rhoME$ in terms of $\HSB$, and explore its consequences.

Consider a system-bath composite held at fixed temperature $T[=1/(\beta k)]$, evolving with the full Hamiltonian $H=H_0+\HSB$ with $H_0\equiv \HS+\HB$. $\HSB$ is taken to be a small perturbation to $H_0$. The $\rhoME$ state is given by Eq.~\eqref{eq:rhoME1}
with $\Lambda$ determined from the constraint $\trB\{\rhoME\}=\rhoS$.
When $\HSB=0$, the ME state is simply the factorized state, 
\begin{equation}
\rhoME_0\equiv\exp(-\Lambda_0-\beta H_0)/\tr{\ldots}=\rhoS\otimes\rhoBth=\cA_0(\rhoS),
\end{equation}
with $\Lambda_0$ determined by $\trB\{\rhoME_0\}=\exp[-\Lambda_0-\beta \HS]/\tr{\ldots}=\rhoS$ so that 
\begin{equation}\label{eq:Lambda0}
\Lambda_0=-\log\rhoS-\beta\HS.
\end{equation}
Here, the natural logarithm $\log(\cdot)$ is understood to be taken on the support of $\rhoS$.
$\Lambda_0$ is a system-only operator, and is unique up to an additive constant of no consequence: Any constant added to the right-hand-side of Eq.~\eqref{eq:Lambda0} is removed by the $\tr{\ldots}$ normalization factor in $\rhoME_0$. A similar remark applies to $\Lambda$ and $\rhoME$.

Since $\HSB$ is a weak perturbation to $H_0$, one expects $\rhoME$ to be close to $\rhoME_0$; correspondingly, $\Lambda$ is expected to be close to $\Lambda_0$. We denote their difference by $\delta\Lambda\equiv \Lambda-\Lambda_0$, a quantity expected to be small. One can obtain an explicit formula for $\delta\Lambda$, accurate to linear order in $\beta\HSB$, which we phrase below as a proposition for clarity.

\begin{proposition*}
The ME state for weak system-bath coupling specified by the interaction Hamiltonian $\HSB$ at fixed temperature $T=1/(\beta k)$ is 
\begin{equation}\label{eq:rhoME2}
\rhoME=\exp(-\Lambda_0-\delta\Lambda-\beta H)/\textup{tr}\{\ldots\},
\end{equation}
with $\Lambda_0=-\log\rhoS-\beta\HS$ [Eq.~\eqref{eq:Lambda0}], and
\begin{equation}\label{eq:dL}
\delta\Lambda=- \textup{tr}_\textup{B}\{\rho_\textup{B}^{(\textup{th})}\beta\HSB\},
\end{equation}
to linear order in $\beta\HSB$.
\end{proposition*}
\begin{proof}
Define $f(O)\equiv\upe^{-O}/\tr{\upe^{-O}}$ for any operator $O$. Let $K\equiv\Lambda_0+\beta H_0$. Then, $\rhoME_0=f(K)$ and $\rhoME=f(\Lambda+\beta H)=f(K+\delta K)$, with $\delta K\equiv \delta\Lambda+\beta \HSB$, taken as small compared with $K$. Let $\rhoME-\rhoME_0=f(K+\delta K)-f(K)\equiv\delta f$, evaluated as a linear response to a variation in $K$ as
\begin{eqnarray}\label{eq:df}
\delta f&=&\frac{1}{\tr{\text{e}^{-K}}}{\left[\delta(\text{e}^{-K})-f(K)\tr{\delta(\text{e}^{-K})}\right]}\\
&=&-f(K)\!{\left[\int_0^1\!\!\!\!\upd\gamma\,\upe^{\gamma K}\delta K\upe^{-\gamma K}-\tr{\delta Kf(K)}\right]}\!.~\nonumber
\end{eqnarray}
The second line comes from noting that $\delta(\upe^{-K})=-\int_0^1 \text{d}\gamma ~\upe^{-\gamma K}\,(\delta K)\,\upe^{-(1-\gamma) K}$ from standard perturbation theory, and consequently, $\tr{\delta(\upe^{-K})}=-\tr{\upe^{-K}\,\delta K}$. The constraints $\trB\{\rhoME\}=\rhoS=\trB\{\rhoME_0\}$ that determine $\Lambda$ and $\Lambda_0$, respectively, imply $\trB\{\delta f\}=0$. Note that (i) $\rhoME_0=f(K)$ can be written as $\rhoS\otimes\rhoBth$; (ii) $\rhoBth$ and $K$ commute; and (iii) $\Lambda$ and $\Lambda_0$, and hence $\delta\Lambda$, are system-only operators. Using these facts, straightforward simplification of $\trB\{\delta f\}=0$ upon insertion of Eq.~\eqref{eq:df} permits the identification of its solution: One requires the relationship $\delta\Lambda=-\trB\{\rhoBth\beta\HSB\}$, which is Eq.~\eqref{eq:dL}. This holds to linear order in $\beta\HSB$, due to the linear-response consideration of $\delta f$, which ignores quadratic terms with $(\delta K)^2=(\delta\Lambda+\beta\HSB)^2\sim(\beta\HSB)^2$.
\end{proof}

Two remarks are in order. First, if the system is prepared in a pure state $|\psi_\text{S}\rangle$, then, the joint system-bath ME state is a tensor-product state by construction. This is enforced by the constraint $\trB\{\rhoME\}=\trB\{\rhoME_0\}=\ket{\psi_\text{S}}\bra{\psi_\text{S}}$, via $\trB\{\delta f\}=0$. This must also be true, accurate to $O(\beta H_\text{SB})$, from direct consideration of Eq.~\eqref{eq:rhoME2} and Eq.~\eqref{eq:dL}. The example in Section \ref{subsec:JC} verifies this explicitly for the Jaynes-Cummings model. Second, note that inserting Eq.~\eqref{eq:dL} for $\delta\Lambda$ into $\rhoME$ of Eq.~\eqref{eq:rhoME2} yields a positive, trace-unity operator for $\rhoME$. More precisely, one can define an assignment map,
\begin{equation}
\mapA(\rhoS)=\exp[\log\rhoS-\delta\Lambda-\beta(\HB+\HSB)]/\tr{\ldots},
\end{equation}
with $\delta\Lambda$ given by Eq.~\eqref{eq:dL}.
Then, the ME state for system state $\rhoS$ under the weak coupling, fixed temperature situation, is simply
\begin{equation}
\rhoME=\mapA(\rhoS)+\ordO\bigl[(\beta\HSB)^2\bigr].
\end{equation}
$\mapA$ hence takes an input system state and outputs a system-bath state approximating the ME state. One can see explicitly that it is a nonlinear map on $\rhoS$, except in the trivial case when $\HSB=0$. For $\HSB=0$, $\mapA(\rhoS)=\rhoS\otimes\rhoBth=\cA_0(\rhoS)$, and the linearity in the input argument is obvious. For the rest of this paper, we will use $\rhoME$ and $\mapA$ interchangeably, with the understanding that the equivalence is up to $O(\beta\HSB)$.

\subsection{\label{sec:WeakB}Deviation of $\rhoME$ from uncorrelated state}
That $\rhoME$ is close to $\rhoME_0$ when $\HSB$ is weak can be quantified using the Proposition when the Hamiltonians involved are bounded operators:
\begin{corollary}\label{cor:1}
For weak system-bath coupling specified by the interaction Hamiltonian $\HSB$ at fixed temperature $T=1/(\beta k)$,
\begin{equation}\label{eq:norm_ineq}
\Vert\mapA(\rhoS)-\rhoME_0\Vert\leq 4\beta\Vert\HSB\Vert,
\end{equation}
where $\Vert O\Vert\equiv\tr{\sqrt{O^\dagger O}}$ denotes the trace norm.
\end{corollary}
\begin{proof}
We return to Eq.~\eqref{eq:df} and rewrite $\delta f$ as a sum of two terms for the two additive parts of $\delta K$: $\delta f=g(\delta\Lambda)+g(\beta\HSB)$, where $g(\cdot)\equiv-f(K){\bigl[\int_0^1 \upd\gamma ~\upe^{\gamma K}(\cdot)\,\upe^{-\gamma K}-\tr{(\cdot) f(K)}\bigr]}$.
Since the trace norm is subadditive and submultiplicative, and $\Vert f(K)\Vert=1$, we have $\Vert g(O)\Vert\leq 2\Vert O\Vert$ for any operator $O$. This immediately gives $\Vert g(\beta\HSB)\Vert\leq 2\beta\Vert\HSB\Vert$. Furthermore, $\Vert g(\delta\Lambda)\Vert\leq 2\Vert\delta\Lambda\Vert=2\Vert \trB\{\rhoBth\beta\HSB\} \Vert \leq 2 \beta\Vert H_{\text{SB}}\Vert$. Thus, $\Vert \delta f\Vert =\Vert g(\delta\Lambda) +g(\beta\HSB)\Vert 
\leq \Vert g(\delta\Lambda)\Vert+ \Vert g(\beta\HSB) \Vert\leq 4\beta\Vert \HSB\Vert$, which is Eq.~\eqref{eq:norm_ineq}.
\end{proof}

This tells us that approximating the initial system-bath state as $\rhoS\otimes\rhoBth$ is sufficient as long as we are unable to detect effects of order $\beta\Vert\HSB\Vert$. It is also a statement that $\rhoME$ contains only weak correlations between system and bath, and hence any effects that rely on initial system-bath correlations are small. We see an example of this in the next subsection.

\subsection{Deviation from complete positivity}

An immediate consequence of initial correlations between system and bath is the possibility of non-CP effects in the subsequent system dynamics. A positive map $\Phi$ is one that takes positive semidefinite operators to positive semidefinite operators (maps states to states). A CP map further requires all extended maps $\Phi\otimes\ I_d$ to be positive, for any $d$ and for $I_d$ as the identity map acting on an ancillary system of dimension $d$. The CP property ensures that a composite system-ancilla state, generally with entanglement between system and ancilla, remains positive under $\Phi$ that acts on the system only.

The question of complete positivity comes up naturally when describing system-only dynamics from joint system-bath evolution. If the system is initially uncorrelated with the bath, the map $\Phi_t$ that takes $\rhoS=\trB\{\rhoSB(t=0)\}$ to $\rhoS(t)\equiv\trB\{\rhoSB(t)\}=\trB\{U(t,0)\rhoSB(0)U(t,0)^\dagger\}$ is CP. Here, $U(t,0)$ is the joint system-bath evolution operator.
When there is initial correlation between the system and the bath, as is the case when using the ME state as the starting state $\rhoSB(0)$, $\Phi_t$ may no longer be CP. More precisely, the map $\Phi_t$ will be CP if the initial system-bath state has zero discord \cite{Rodriguez-Rosario08}; otherwise, non-CP effects can emerge. Any observation of non-CP effects can hence be used as a witness for the presence of nonzero discord in the initial system-bath state \cite{Gessner11, Gessner14, Modi12, Ringbauer14}.
One might also imagine exploiting the effects that can arise when the CP restriction is lifted, as explored in many recent papers. Using our expression for $\rhoME$, one can see how strong these effects can be. 

A natural ``witness" for non-CP dynamics is the trace distance between two system states, a measure of distinguishability between the two states. Under a CP map, the distinguishability is non-increasing (see, for example, \cite{NCBook}); a violation signals non-CP effects. This fact is used to construct a witness for non-Markovian dynamics for evolution from time $t$ to an infinitesimally later time $t+\delta t$ \cite{Breuer09}; here, we compare only to initial time $t=0$, to detect the effects of a correlated initial state \cite{Laine2010}.
The following emerges immediately from the Proposition:
\begin{corollary}
Consider the situation of weak system-bath coupling as in the Proposition. Let $\Phi_t$ be a map on system-only states such that $\Phi_t(\cdot)=\trB\{U(t,0)\mapA(\cdot)U(t,0)^\dagger\}$ for some joint evolution $U$. Let $\rhoS(t)\equiv\Phi_t(\rhoS)$. For any two system states $\rhoS$ and $\rhoS'$, let $\Delta(\rhoS,\rhoS';t)\equiv \Vert\rhoS(t)-\rhoS'(t)\Vert-\Vert\rhoS-\rhoS'\Vert$. Then, 
\begin{eqnarray}\label{eq:trdist}
\max_{\rhoS,\rhoS'}\Delta(\rhoS,\rhoS';t)\leq 8\beta\Vert\HSB\Vert
\end{eqnarray}
for any time $t>0$.
\end{corollary}
\begin{proof}
Consider any two system-only states $\rhoS$ and $\rhoS'$. 
\begin{eqnarray}
&&\Vert \rhoS(t)-\rhoS'(t)\Vert=\Vert\Phi_t(\rhoS)-\Phi_t(\rhoS')\Vert\nonumber\\
&=&\Vert\trB\{U(t,0)[\mapA(\rhoS)-\mapA(\rhoS')]U(t,0)^\dagger\}\Vert\nonumber\\
&\leq &\Vert\mapA(\rhoS)-\mapA(\rhoS')\Vert
\nonumber\\
&=&\Vert\mapA(\rhoS)-\rhoS\otimes\rhoBth+\rhoS\otimes\rhoBth\nonumber\\
&&\qquad-\rhoS'\otimes\rhoBth+\rhoS'\otimes\rhoBth-\mapA(\rhoS')\Vert\nonumber\\
&\leq&\Vert\mapA(\rhoS)-\rhoME_0\Vert+\Vert\rhoS-\rhoS'\Vert\nonumber\\
&&\qquad+\Vert(\rho')^{(\mathrm{ME})}_0-\mapA(\rhoS')\Vert\nonumber\\
&\leq&8\beta\Vert\HSB\Vert+\Vert\rhoS-\rhoS'\Vert.
\end{eqnarray}
Here, $(\rho')^{(\mathrm{ME})}_0\equiv\cA_0(\rhoS')= \rhoS'\otimes \rhoBth,$ and we have used Eq.~\eqref{eq:norm_ineq} in the last line. We thus have that $\Delta(\rhoS,\rhoS';t)\leq 8\beta\Vert\HSB\Vert$. This bound is independent of the choice of $\rhoS$ and $\rhoS'$, so Eq.~\eqref{eq:trdist} holds.
\end{proof}

\subsection{Illustrating examples}
\subsubsection{The Jaynes-Cummings Hamiltonian}\label{subsec:JC}
As a simple example, we work out the ME state for the Jaynes-Cummings Hamiltonian, where a spin-1/2 particle (the system) is coupled to a single mode of light (the bath),
\begin{equation}
H=\underbrace{\frac{1}{2}\hbar \omega_a\sigma_z}_{H_\text{S}}+\underbrace{\hbar\omega a^\dagger a}_{H_\text{B}}+\underbrace{J(\sigma_-a^\dagger+\sigma_+ a)}_{H_{\text{SB}}}.
\end{equation}
We denote spin up(down) by $\uparrow$($\downarrow$), $\sigma_i,i=x,y,z$, are the Pauli operators, $\sigma_{\pm}=\frac{1}{2}(\sigma_x\pm \text{i}\sigma_y)$ are the spin raising ($\sigma_+=\uket\dbra$) and lowering ($\sigma_-=\dket\ubra$) operators, and $a$ and $a^\dagger$ are the photon ladder operators. $\HSB$ describes the absorption and emission of a photon by the spin, with the corresponding change of spin state. By an appropriate choice of phases in the $a$ and $a^\dagger$ operators, or in the spin states, $J$ can be taken as real without any loss of generality. Furthermore, for the split into system and bath to be sensible, $\HSB$ is a weak correction to the free-evolution Hamiltonian, so that $J$ is small compared with the energy scales of $\HS$ and $\HB$.

The bath-only thermal state is $\rhoBth=\upe^{-\hbar\omega a^\dagger a}/\textrm{tr}\{\ldots\}$. Since both $\textrm{tr}\{\upe^{-\hbar\omega a^\dagger a}a\}$ and $\textrm{tr}\{\upe^{-\hbar\omega a^\dagger a}a^\dagger\}$ are zero, $\delta \Lambda$ of Eq.~\eqref{eq:dL} vanishes, so that $\Lambda=\Lambda_0=-\log\rhoS-\frac{1}{2}\beta\hbar\omega_a\sigma_z$. The ME state is thus $\rhoME=\exp{\left[\log\rhoS-\beta (\HB+\HSB)\right]}/\tr{\ldots}$. A quick check of the formula yields $\rhoME=\rhoS\otimes\rhoBth=\rhoME_0$ when $\HSB=0$, as it should be.

When $\HSB$ is nonzero, the correction to $\rhoME_0$ given by our ME approach can be worked out explicitly. In this case, the ME state can be written as $\rhoME=\mathrm{exp}(-K-\delta K)/\tr{\cdots}$, with $K=-\log\rhoS+\beta\hbar\omega a^\dagger a$ and $\delta K=\beta J(\sigma_-a^\dagger+\sigma_+ a)$. Note that $\rhoME_0=\exp(-K)/\tr{\ldots}$ in this notation. Using the Bloch-sphere representation for the state of the two-dimensional system, we can write $\rhoS=\frac{1}{2}(\boldsymbol{1}+\vec{s}\cdot\vec{\sigma})=\upe^{-\vec{\lambda}\cdot\vec{\sigma}}/\tr{\ldots}$ with $\vec{\lambda}=\lambda\vec{e}_s$,  where $\vec{e}_s\equiv\vec s/|\vec s|$, and $\lambda\equiv-\tanh^{-1}(s)$. Then, $K=\vec\lambda\cdot\vec \sigma+\beta\hbar\omega a^\dagger a$, where we have dropped an additive constant with no effect on $\rhoME$.
Let $B$ be an operator such that $\rhoME=\upe^{-B}\rhoME_0\upe^{B} = \upe^{-B}\upe^{-K}\upe^{B}/\tr{\cdots}$; the existence of such a $B$ is justified below. The operator identity $\mathrm{exp}(\upe^{-B}K\upe^{B})=\upe^{-B}\upe^{K}\upe^{B}$ tells us to set $\upe^{-B}K\upe^{B}=K+\delta K=K+\beta J(\sigma_-a^\dagger+\sigma_+ a)$. Since $\beta J$ is a small parameter, this last expression suggests that $\upe^B$ is close to the identity, or equivalently, that $B$ is small. Approximating $\upe^B\simeq 1+B$, we require $[K,B]\simeq\beta J(\sigma_-a^\dagger + \sigma_+a)$. The ansatz $B=\vec\alpha_+\cdot\vec\sigma a^\dagger+\vec\alpha_-\cdot\vec\sigma a+u_-\sigma_-a^\dagger + u_+\sigma_+a$ satisfies this requirement when
\begin{align}
u_\pm&=\pm J/\hbar\omega\nonumber\\
\textrm{and}\quad \vec\alpha_\pm&=\pm\upi2(J/\hbar\omega){M_\pm}^{-1}(\vec\lambda\times\vec\mu_\mp), \label{eq:alpha}
\end{align}
where $\vec\mu_\pm=\tfrac{1}{2}(1,\pm\upi,0)^\mathrm{T}$ so that $\sigma_\pm=\vec\mu_\pm\cdot\vec\sigma$, and $M_\pm$ are the matrices
\begin{equation}
M_\pm={\left(\begin{array}{ccc}
\pm\beta\hbar\omega&-\upi 2\lambda_z&\upi 2\lambda_y\\
\upi 2\lambda_z&\pm\beta\hbar\omega&-\upi 2\lambda_x\\
-\upi 2\lambda_y&\upi 2\lambda_x&\pm\beta\hbar\omega
\end{array}\right)}.
\end{equation}
$M_\pm$ are invertible except in accidental situations [$\det(M)\sim(\beta\hbar\omega)^3+(\lambda \textrm{ terms})$, which is never zero unless $\lambda\sim\beta\hbar\omega$].

Let us examine two special cases. First, for $\rhoS=\tfrac{1}{2}\boldsymbol{1}$, $\vec s$---and hence $\vec \lambda$---is zero. Then, $B=(J/\hbar\omega)(\sigma_-a^\dagger-\sigma_+a)$ and $\rhoME=\rhoME_0-[B,\rhoME_0]=\rhoME_0{\left\{1-(J/\hbar\omega){\left[\sigma_-a^\dagger (\upe^{\beta\hbar\omega}-1)-\sigma_+a(\upe^{-\beta\hbar\omega}-1)\right]}\right\}}$,   to linear order in $J/\hbar\omega$. As expected, the ME state differs from $\rhoME_0$ by a term linear in $\beta J$. If the experiment is sensitive to deviations of this size, the use of an initially uncorrelated state becomes inappropriate. 
Next, consider a pure state $\rhoS=\ket{\psi_\text{S}}\bra{\psi_\text{S}}$, with Bloch vector $\vec s$ of unit length. As mentioned in Sec.~\ref{sec:weakCoup}, $\rhoME$ in Eq. \eqref{eq:rhoME2} should be a product state with $\trB\{\rhoME\}=\ket{\psi_\text{S}}\bra{\psi_\text{S}}$. Solving for  Eq.~\eqref{eq:alpha}, taking $|\vec s|\rightarrow1$, or equivalently, $\lambda\rightarrow \infty$, this gives $\vec{\alpha}_\pm=-\tfrac{J}{2\hbar\omega}(\pm1,\mathrm{i},0)^\mathrm{T}$. Consequently, $B=0$, and so $\rhoME=\rhoME_0$, indeed a product state.

\subsubsection{The central spin model}
As a second example, we consider what is known as the ``central spin model": a central spin (the system) coupled to a number of other spins (the bath). Here, a ``spin" can be broadly understood as a two-level system or ``qubits", e.g., spin-1/2 particles in a bias magnetic field, atoms with two energetically relevant levels, etc. The central spin model applies in a large variety of physical situations;  see, for example, the review article Ref.~\cite{Prokofev00}. In particular, for the development of quantum devices, in which the study of the effects of the bath is necessary for good control, the central-spin model (and its extensions to structured spin baths) has been found to be the appropriate description of noise for many solid-state architectures; see, for example, Ref.~\cite{Muller15} for superconducting qubits, Ref.~\cite{Hall14} for NV centers, and Ref.~\cite{Brownnutt15} for anomalous heating issues in surface traps for ions.
These solid-state devices are often temperature controlled---e.g., the emission/absorption spectral line is temperature sensitive and has to be stabilized---and, through material characterization, one often has information on the most relevant two-level degrees of freedom forming the spin bath. These are the typically available macroscopic information.

For simplicity of computation, let us assume all spins (one system spin plus $N$ bath spins) have the same properties, namely, the same energy splitting (of size $\hbar\omega$), and pairwise interaction (along the quantization axis) of the same strength between all spins. The Hamiltonian is $H=H_\text{S}+H_\text{B}+H_\text{SB}$, with $H_\text{S}=\sigma_z^{\text{S}}$, $H_\text{B}=\sum_{i=1}^N \sigma_z^{(i)}+g\sum_{i<j}\sigma_z^{(i)} \sigma_z^{(j)}$, and $H_\text{SB}=g\sum_{i=1}^N  \sigma_z^{\text{S}} \sigma_z^{(i)}$, all measured in units of $\tfrac{1}{2}\hbar\omega$, and $\sigma_z^{(i)}$ is the Pauli-$Z$ operator acting on the $i$th bath spin. For weak interaction, $g\ll1$. One can again apply the same recipe in the Proposition to work out the ME state for the fixed temperature situation. Here we focus instead on illustrating the two corollaries.

In Fig.~1, we plot the quantities on the left- and right-hand-sides of Eq.~\eqref{eq:norm_ineq} in Corollary \ref{cor:1}, for fixed $g$ and varying the number of bath spins $N$. One sees that as $N$ increases, the difference between ME state and $\rho_0^{\text{(ME)}}$ increases. Meanwhile, the bound given by Eq.~\eqref{eq:norm_ineq} increases faster. This is not surprising, since our bound is a uniform bound that holds for all situations. Figure 2 shows the same quantities as in Fig.~1, but now with fixed $N=1$ and varying $g$. As the $g$ decreases, we see the two states approaching each other. Note that the blue circles and red asterisks differ by a $g$-independent constant order of magnitude, consistent with a scaling of both quantities by the same power of $\Vert\beta\HSB\Vert$.

\begin{figure}[t!]
\centering
\includegraphics[width=0.43\textwidth]{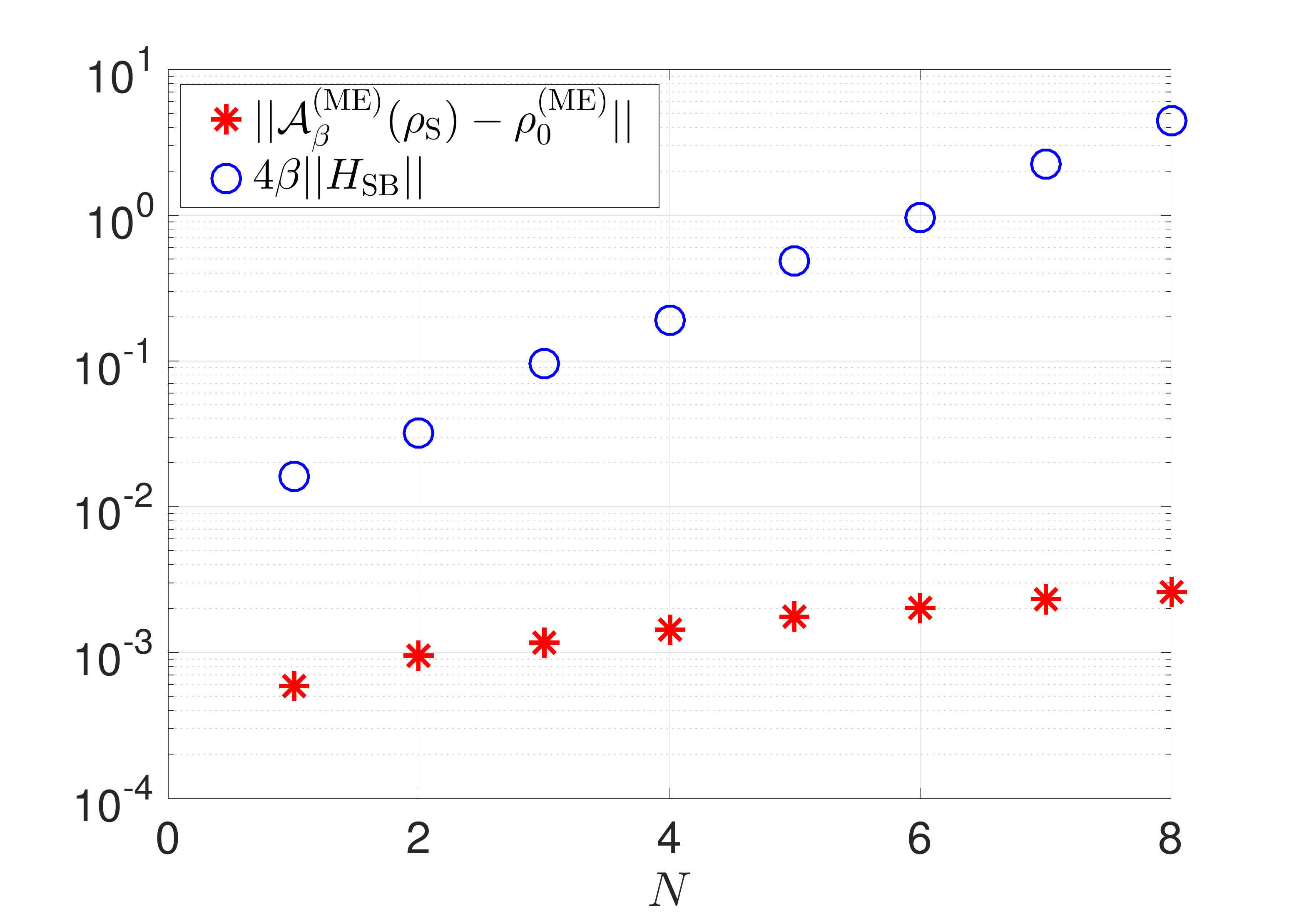}
\caption{\label{Fig15a}%
A central spin coupled to $N$ bath spins: The left-(red asterisks) and right--hand-sides (blue circles) of Eq.~(16) as a function of the number $N$ of bath spins in the central spin model. $g$ is fixed at $10^{-3}$. $\rho_\text{S}$ is a random system state; the features of the graph do not vary much with different choices of $\rho_S$.} 

\end{figure}

\begin{figure}[t!]
\centering
\includegraphics[width=0.43\textwidth]{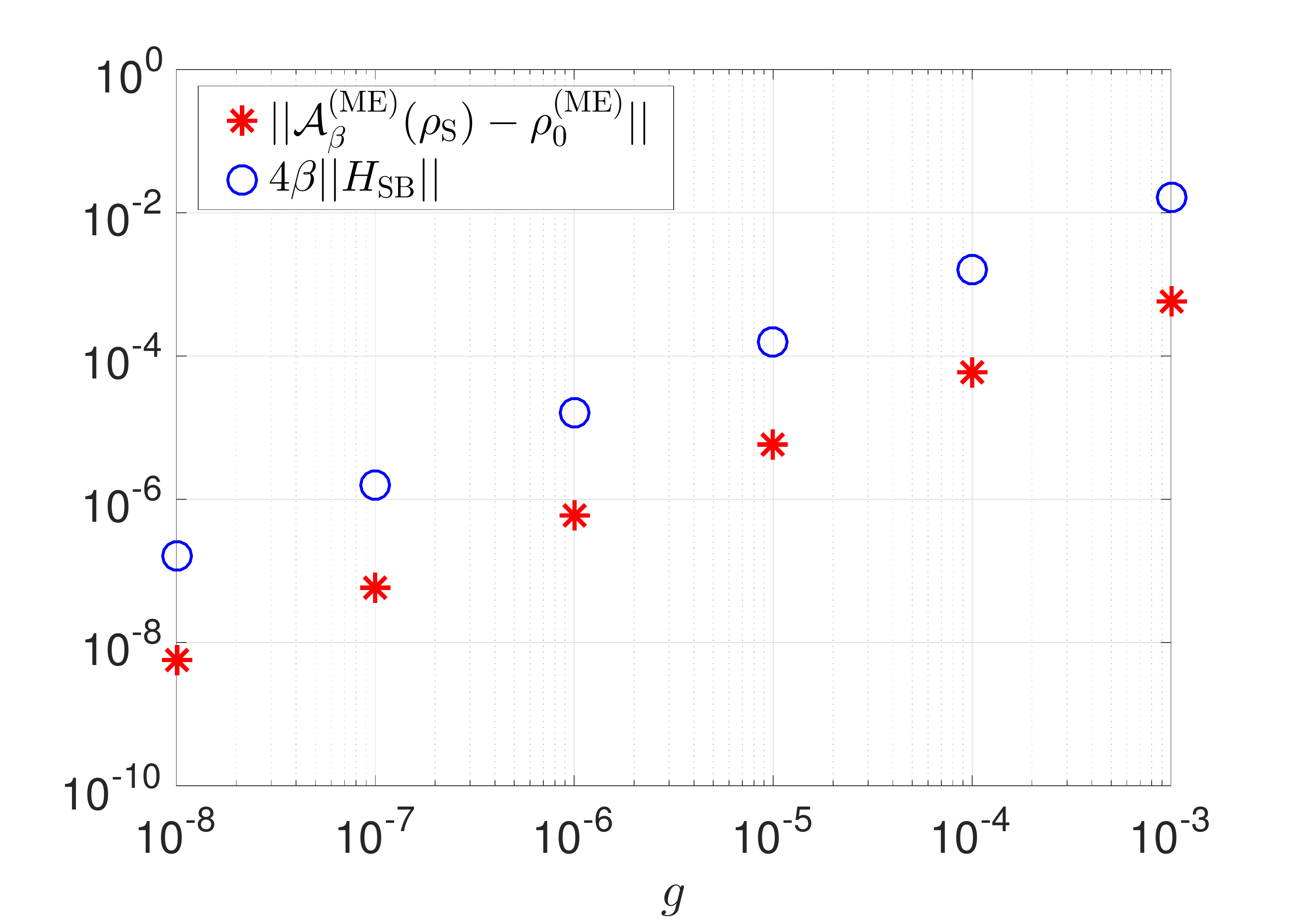}
\caption{\label{Fig15b}%
A central spin coupled to one bath spin: The left-(red asterisks) and right-(blue circles) hand side of Eq.~(16) as a function of the interaction strength $g$ in the central spin model. As in Fig.~\ref{Fig15a}, $\rho_\text{S}$ is a random system state; the features of the graph do not vary much with different choices of $\rho_S$.} 

\end{figure}

\begin{figure}[h!]
\centering
\includegraphics[trim=5mm 2mm 5mm 5mm, clip, width=\columnwidth]{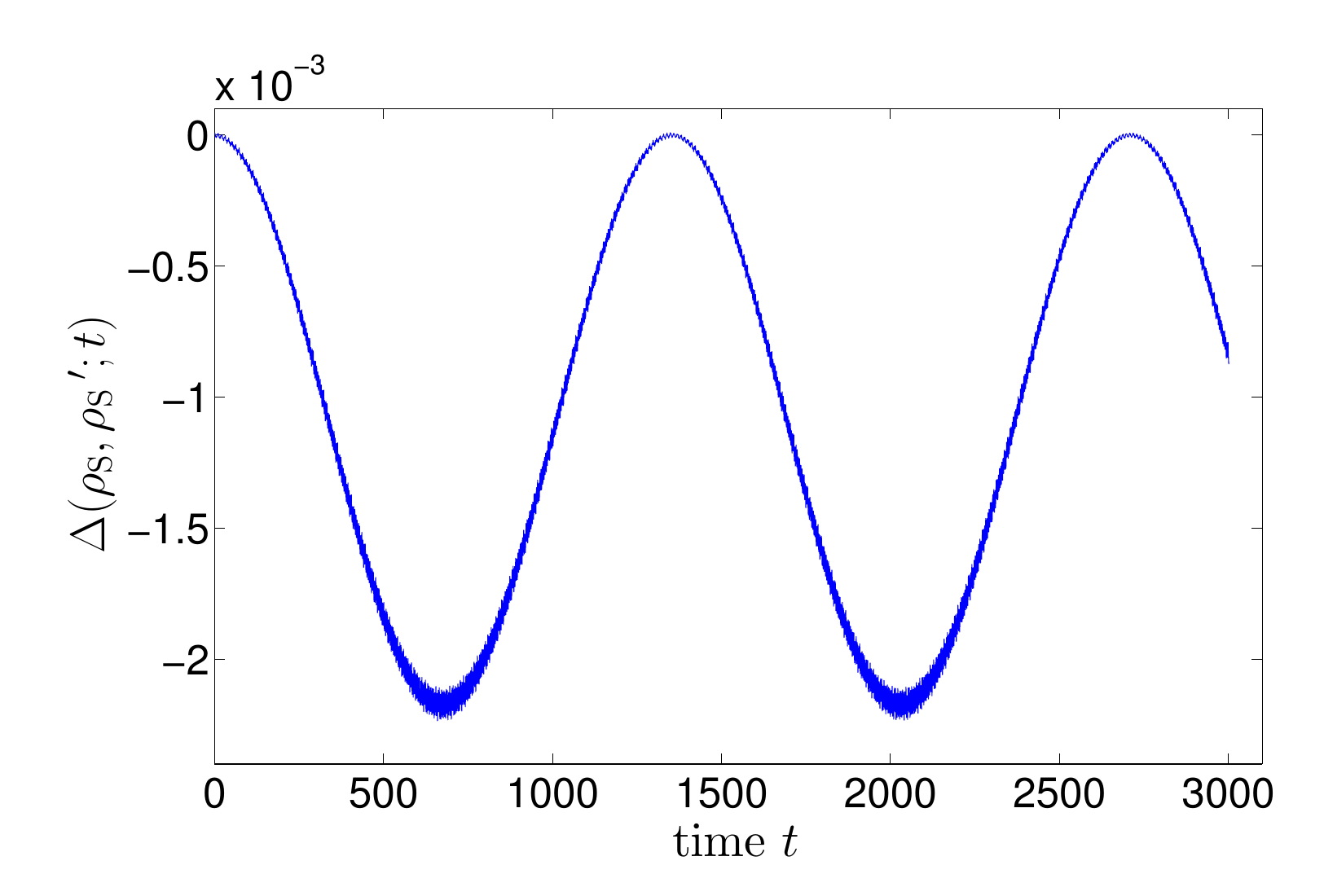}
\caption{\label{trace_distance}%
A central spin coupled to one bath spin: The full Hamiltonian $H$ is randomly chosen with the requirement that $\HSB$ is 1\% in norm compared to $\HS$ and $\HB$; the two initial system states are also randomly chosen. The plot shows a typical graph for $\Delta(\rhoS,\rhoS';t)$. One sees that $\Delta$ hardly rises above $0$ ($t$ in units of the maximum eigenenergy of the system divided by $\hbar$). While Eq.~\eqref{eq:trdist} provides an upper bound for $\Delta$, we note that this bound is not anywhere close to being achieved for any of the cases in our numerical simulation of 3000 randomly chosen Hamiltonians, each with 100 pairs of system initial states.
}
\end{figure}

As an illustration of Corollary 2, we study a similar example of a central spin coupled to a bath spin, but with a randomly chosen Hamiltonian $H$. Fig.~\ref{trace_distance} shows the typical behavior of $\Delta(\rhoS,\rhoS';t)$ for a randomly chosen pair of initial system states. Eq.~\eqref{eq:trdist} gives the maximal violation of the non-increasing property of distinguishability under a CP map for any two initial system states.
It is a precise statement of how initial system-bath correlations in the weak coupling situation have only weak consequences: One needs to detect effects of order $\beta\Vert\HSB\Vert$ to see a deviation from dynamics starting from a factorized initial state.

\section{Discussion}

The system tomography process reliably verifies that the system is prepared in a certain state $\rhoS$. For each $\rhoS$, there are infinitely many system-bath states $\rho$ consistent with $\trB\{\rho\}=\rhoS$. Macroscopic constraints on the system and the bath further restrict $\rho$, giving $\Sigma$ as the set of system-bath states consistent with the available information. The inability to measure microscopic features of the bath translates into our inability to single out---based only on experimental data---a particular $\rho$ from $\Sigma$. 
In the preceding sections, we employed the additional principle of maximum entropy to identify a unique $\rho$ from $\Sigma$, as a quantitative statement of our lack of knowledge. There, any knowledge of the system-bath state is accounted for via the use of Lagrange multipliers. 

We note that our knowledge and partial ignorance of the system-bath state can alternatively be incorporated using Bayesian methods, \emph{provided} one knows how to choose the underlying weight properly. Concretely, one reports the \emph{Bayesian mean} well known from Bayesian statistical inference: the average over all $\rho$ in $\Sigma$, $\overline{\rho}=\int_{\Sigma}(\text{d}\rho)\,\rho\,$. Here, $(\upd \rho)$ is a normalized weight density---the prior---over $\Sigma$, chosen to represent our available information. For example, consider Case 3 of Sec.~\ref{sec:3} where no information other than the system tomography data is available, so that $\Sigma$ comprises all states such that $\trB\{\rho\}=\rhoS$. A plausible expression of this complete ignorance of the bath state is the Haar average construction,
$\overline{\rho}=\int(\text{d}U_\text{B})(\boldsymbol{1}_\text{S}\otimes U_\text{B})|\psi\rangle_{\text{SB}}\langle\psi|(\boldsymbol{1}_\text{S}\otimes U_\text{B}^\dagger)$,
where $\int(\text{d}U_\text{B})$ averages over unitaries on the bath according to the Haar measure, with $|\psi\rangle_\mathrm{SB}$ a fiducial state such that $\trB\{|\psi\rangle_\mathrm{SB}\langle\psi|\}=\rhoS$. This gives $\overline{\rho}=\rhoS\otimes \boldsymbol{1}_\text{B}/d_\text{B}$, identical to that obtained earlier using the ME reasoning. However, we note that a different choice of ``uniform" weight over the bath states generally yields a $\overline{\rho}$ different from the corresponding $\rhoME$, and care should be taken in the choice of $(\upd\rho)$ for an apt expression of our ignorance of the bath state. We note also the work in Refs.~\cite{Paris07,Paris14} based on minimizing the relative entropy to a state $\tau$ which is treated as the prior; their procedure yields the ME state for the typical experimental situation of system-bath composite held at fixed temperature, \emph{provided} $\tau$ is the joint system-bath thermal state.

That the initial correlations between the system and bath are weak, as in Sec.~\ref{sec:WeakB}, is a consequence of the weak $\HSB$. It is tempting to extend this to the case of strong system-bath coupling. Certainly, there are four conditions under which the correlations between system and bath will clearly remain weak even for strong $\HSB$: (a) the state $\rhoS$ is nearly pure, so that any correlations with the bath state must necessarily be weak, regardless of $\HSB$; (b) the system is prepared separately from the bath, so that the product $\rhoS\otimes\rhoBth$ is the appropriate initial state (Case 1 of Sec.~\ref{sec:3}); (c) there is no knowledge other than the system tomographic constraint, such that the ME state is the product state $\rhoS\otimes \frac{1}{d_\text{B}}\boldsymbol{1}_\mathrm{B}$ (Case 3 of Sec.~\ref{sec:3}); (d) there is a temperature constraint on the system and bath, but the temperature is high enough such that many system-bath states are involved in the thermal mixture that follows from the ME state. 

When none of the four conditions above hold, and $\HSB$ is strong, the weak system-bath correlations conclusion is unlikely to be the right one. In this case, in the absence of the system tomographic constraint as enforced by $\Lambda$, the system-bath state will have a large weight on the ground state of the full Hamiltonian $H$, which includes the effects of the strong $\HSB$. With the state preparation constraint [but not satisfying Conditions (a) or (b) above], the system-bath state can still generally have strong correlations. One way of treating this would be that of Ref.~\cite{Chaudhry13} and others, of treating the state preparation as an operation done on the possibly strongly correlated joint thermal state. However, in such a situation, the split into system and bath is called into question. For quantum information processing tasks, in particular, one cannot hope for good control if one chooses to work with a system strongly coupled to an uncontrollable part.

So far, we have been considering a single joint system-bath state consistent with our incomplete knowledge. 
Such ``point estimators" are useful whenever one prefers a single ``best guess", for use in future calculations. One might instead be interested in properties of $\Sigma$ itself. For example, suppose we only have knowledge about the system state $\rho_\text{S}$, and one is concerned with the question: How far can any state in the set $\Sigma$ be from a tensor-product state? In this case, we can make use of the conditional entropy $\cS(\textrm{B}|\textrm{S})=\cS(\textrm{B},\textrm{S})-\cS(\textrm{S})$, where $\cS(\textrm{B},\textrm{S})=\cS(\rhoSB)$ is the joint system-bath entropy. The conditional entropy captures how much of the bath state we know, given $\rhoS$. A tensor-product state $\rhoS\otimes\rhoB$ gives $\cS(\textrm{B}|\textrm{S})=\cS(\rhoB)$, which is bounded by the maximal value of $\log d_\textrm{B}$ (attained by the state $\rhoS\otimes \boldsymbol{1}_\text{B}/d_\text{B}\in\Sigma$), and the minimal value of $0$ (attained by $\rhoS\otimes\rhoB\in\Sigma$ when $\rhoB$ is pure). Any negative value of $\cS(\textrm{B}|\textrm{S})$ hence indicates correlations (classical or quantum) between the system and bath. The lower bound is attained by the most correlated state in $\Sigma$, namely a purification of $\rhoS$, so that $\min_{\rhoSB\in\Sigma}\cS(\textrm{B}|\textrm{S})=-\cS(\rhoS)$.

In fact, the maximum entropy state that we look for in Sec.~\ref{sec:3} also maximizes the conditional entropy $\cS(\textrm{B}|\textrm{S})$ as the entropy of the system state is fixed by the tomography. It thus comes as little surprise that the maximum entropy state is uncorrelated or almost uncorrelated between the system and the bath. On the other hand, if one has at hand a situation requiring a conservative estimate of the maximal amount of correlation consistent with the given data---e.g., in the case of security analysis of quantum key distribution protocols---then a purification $|\psi\rangle_\text{SB}$ of $\rho_\text{S}$ in $\Sigma$ can be used. However, one should remember that such a pure state can only serve to derive bounds for the ``worst-case scenario'' and should not be misconstrued as a description of the actual system-bath state. The lack of control on the bath should translate into the lack of our ability to single out a particular $|\psi\rangle_\text{SB}$.

\section{Conclusion}
The initial joint system-bath state is not a mysterious quantity that is unknown or arbitrary. Instead, one should take proper account of our knowledge of the system-bath composite, encompassing all measurements already taken to characterize the experimental circumstances. This knowledge gives us a good handle on the system-bath state through the use of ME ideas from quantum tomography. 
In this article, we derived an expression for the ME state in the case of weak system-bath coupling, and showed that it is close to a factorized state.
The rationale for the ME strategy lies with the by-definition uncontrollable nature of the bath: Maximizing the entropy is a quantitative statement of that uncontrollability. Any observed deviation from the predictions of the ME joint state should thus be viewed as indicative of the presence of some microscopically controlled, repeatedly identically prepared degrees of freedom in the bath that one failed to identify. Discovery of these hidden controlled degrees of freedom would allow for the possibility of exploiting them for better and more creative system control. Otherwise, the ME state provides a good description of the initial system-bath situation, with weak correlations in typical situations of experimental interest.

\acknowledgments
We are grateful for insightful discussions with Berge Englert and for his comments on the manuscript. We also thank an anonymous referee for the AQIS2015 conference for helpful suggestions to an early draft of this manuscript.
This work is funded by the Singapore Ministry of Education and the National Research Foundation of Singapore, and by Yale-NUS College (through grant number IG14-LR001, and the Yale-NUS start-up grant).

%%%%%%%%%%%%%%%%%%%%%%%%%%%%%%%%%%%%%%%%%%%%%%

\end{document}